\documentclass[conference]{IEEEtran}
\IEEEoverridecommandlockouts
\usepackage{cite}
\usepackage{amsmath,amssymb,amsfonts,amsthm}
\usepackage{algorithmic}
\usepackage{graphicx}
\usepackage{textcomp}
\usepackage{xcolor}
\def\BibTeX{{\rm B\kern-.05em{\sc i\kern-.025em b}\kern-.08em
    T\kern-.1667em\lower.7ex\hbox{E}\kern-.125emX}}
\usepackage{mathtools,leftidx}

\newtheorem{prop}{Proposition}
\newtheorem{theorem}{Theorem}
\newtheorem{construction}{Construction}
\newtheorem{definition}{Definition}
\newtheorem{example}{Example}
\newtheorem*{remark*}{Remark}
\usepackage{quotmark}
\usepackage{csquotes}

\begin{document}

\title{On cyclic LRC codes that are also LCD codes}

\author{
	\IEEEauthorblockN{Charul Rajput }
	\IEEEauthorblockA{\textit{Department of Mathematics}\\ 
                                                   \textit{ Indian Institute of Technology Roorkee}\\
		   			  Roorkee, Uttarakhand\\
					  India\\
				   	  Email: charul.rajput9@gmail.com\\
\thanks{First author is financially supported by University Grants Commission (UGC), India, under grant code 6405-11-061.}}

  \and
   \IEEEauthorblockN{Maheshanand Bhaintwal}
   \IEEEauthorblockA{\textit{Department of Mathematics}\\ 
                                   \textit{Indian Institute of Technology Roorkee}\\
                   	            Roorkee, Uttarakhand\\
  	                 India\\
  	                 Email: mahesfma@iitr.ac.in}
   \and
   \IEEEauthorblockN{Ramakrishna Bandi}
   \IEEEauthorblockA{\textit{Department of Mathematics}\\ 
                                \textit{Dr. SPM International Institute of}\\
   	                 \textit{Information Technology Naya Raipur}\\
                   	 Chhattisgarh, India\\
   	                 Email: ramakrishna@iiitnr.edu.in}

}

 \IEEEpubid{978-1-7281-9180-5/20/\$31.00 ©2020 IEEEIn} 

\maketitle

\begin{abstract}
Locally recoverable (LRC) codes provide a solution to single node failure in distributed storage systems, where it is a very common problem. On the other hand, linear complementary dual (LCD) codes are useful in fault injections attacks on storage systems. In this paper, we establish a connection  between LRC codes and LCD codes. We derive some conditions on the construction of cyclic LRC codes so that they are also LCD codes. A lower bound on the minimum distance of such codes is determined. Some examples have been given to explain the construction.
\end{abstract}


\section{Introduction}
Due to exponential growth of data, node failures has become a common phenomenon in distributed and cloud storage systems. The most common scenario in this is the failure of a single node. The existing methods of data recovery do not address this problem efficiently. Locally recoverable (LRC) codes provide an efficient solution to this problem. They have been therefore studied quite extensively in the literature in recent years \cite{HCL2007, HL2007, GHSY2012}. A linear code $C$ is said to be an LRC code if the value at any coordinate position of a codeword in $C$ can be determined by accessing a small number of other coordinates of the codeword. The set of  coordinate positions that are used to determine the value at a specified coordinate position is called the recovering set for that coordinate position. If the size of the recovering set for every coordinate in $C$ is less than or equal $r$, where $r$ is a positive integer, then $C$ is said to have locality $r$. Smaller locality represents better local recovery.

Constructions of LRC codes have been studied in a number of recent papers \cite{GHSY2012, SRKV2013, TB2014, GHJY2014, TPD2016}. Tamo and Barg \cite{TB2014} have given a construction of LRC codes which is a generalization of the construction of classical Reed-Solomon (RS) codes. These codes are known as RS-like LRC codes. Bounds on the minimum distance of LRC codes have also been studied by several researchers \cite{GHSY2012, TB2014, CM2015, TBF2016}.  A  generalization of the classical Singleton bound for LRC codes is given in \cite{GHSY2012}. Codes whose minimum distance attains this bound are called optimal LRC codes. Several constructions of optimal LRC codes have been proposed in the literature \cite{SRKV2013, TB2014,  GHJY2014, TPD2016}.   Goparaju et al. \cite{GC2014} have studied binary cyclic LRC codes and have given a construction of LRC codes using the zeros of the generator polynomial.  In \cite{TBGC2015}, Tamo et al. have given a construction of RS-like cyclic LRC codes.

A linear code $C$ whose intersection with its dual is the zero code, i.e., $C \cap C^{\perp}=\{0\}$ or equivalently, $C \oplus C^{\perp}=\mathbb{F}_q^n$, is called a linear complementary dual (LCD) code.
LCD codes were introduced by Massey  \cite{Mas92} in 1992.  These codes are proved useful in communication systems, data storage and cryptography. Recently, in \cite{CG16}, Carlet and Guilley showed that LCD codes helps to secure the communication against the side-channel attacks and fault injection attacks. In \cite{YM94}, Yang and Massey  proved that reversible cyclic codes
over finite fields are LCD codes by providing a necessary and sufficient condition for a cyclic code to have complementary dual. LCD codes are well explored in the literature  \cite{Sen04,  CMTQ18, CMTQP18, LL15}.


LRC codes that are also LCD codes may prove more useful for distributed data storage systems  as they provide protection against node failures and also protection against side-channel attacks and fault injection attacks on such systems.

 In this paper we study cyclic LRC codes that are also LCD codes. We call such codes (cyclic) LRC-LCD codes. We derive some conditions on the construction of cyclic LRC codes so that the codes so constructed are also LCD codes. A lower bound on the minimum distance of these codes is determined.  Some examples have been given to explain the construction.

\section{Preliminaries}
\IEEEpubidadjcol
Let $\mathbb{F}_q$ be a finite field of size $q$, where $q$ is a prime power. An $[n,k,d]$-linear code $C$ is a subspace of $\mathbb{F}_q^n$ with dimension $k$ , where $d$ denotes the minimum Hamming distance.  The dual of $C$, denoted by  $C^{\perp}$, is the set of all vectors of $\mathbb{F}_q^n$ that are orthogonal to all elements of $C$, i.e., $C^{\perp}=\{ x \in \mathbb{F}_q^n \ | \ x \cdot c = 0 \ \forall \ c \in C \}$. $C$ is called a cyclic code if codewords in $C$ are invariant under the cyclic shifts, i.e.,  for every $c=(c_0, c_1, \ldots, c_{n-1}) \in C$, $\sigma(c)=(c_{n-1},c_0, c_1, \ldots, c_{n-2}) \in C$. If each vector $a=(a_0, a_1, \ldots, a_{n-1}) $ of $\mathbb{F}_q^n$corresponds to the  polynomial  $a(x)=a_0+a_1x+\ldots+a_{n-1}x^{n-1}$ of $\mathbb{F}_q[x]$, then a cyclic code $C$ of length $n$ is an ideal of the quotient ring $R=\frac{\mathbb{F}_q[x]}{\langle x^n-1 \rangle}$. Since $R$ is a principal ideal ring, $C$ is generated by a unique monic polynomial $g(x)$, called the generator polynomial of $C$. The zeros of $g(x)$ in its splitting field over $\mathbb{F}_q$ are called the zeros of $C$, and the set of all zeros of $g(x)$ is called the defining set of $C$. The polynomial $h(x)=(x^n-1)/g(x)$ is called a parity-check polynomial of $C$ and its reciprocal polynomial $h^*(x)=x^{\deg h(x)} h(1/x)$ generates the dual code $C^{\perp}$.

A code $C \subseteq \mathbb{F}_q^n$ is said to be a reversible code if for every $(c_0, c_1, \ldots, c_{n-1}) \in C$, the word $(c_{n-1}, c_{n-2}, \ldots, c_0)$ is also in $C$.

Now we give some basic definitions and results on LRC codes and cyclic LCD codes (reversible cyclic codes).

\begin{definition}[LRC codes]
	Let $C$ be an $[n,k,d]$ code over $\mathbb{F}_q$. Then $C$ is said to have locality $r$ if for every $c=(c_1,c_2, \ldots, c_n) \in C$, $1 \le i \le n$, there exists a subset $R_i \subset \{1,2, \ldots, n\} \backslash \{i\}, |R_i| \leq r$, such that $c_i$ can be recovered by using the values of the coordinates associated with $R_i$. The set $R_i$ is called a recovering set of the $i^{\text{\tiny{th}}}$ coordinate and code $C$ is called an $(n,k,r)$ LRC code.
\end{definition}

In \cite{TB2014}, a bound on the minimum distance of LRC codes has been given. If $C$ is an $(n,k,r)$ LRC code then the minimum distance $d$ satisfies, 
\begin{equation}\label{eq1}
d \leq n-k- \left\lceil \frac{k}{r} \right\rceil +2~.
\end{equation}
An LRC code is called an optimal LRC code if its minimum distance attains the above bound.

Let $n$ be a positive integer such that $n|(q-1)$. Let $k$ and $r$ be positive integers such that $(r+1)|n$ and $r|k$. The following result gives a construction of RS-like cyclic LRC codes of length $n$.
\begin{theorem} \cite{TBGC2015} \label{th1}
Let $\alpha$ be a primitive $n$-th root of unity, $\ell$ be an integer with  $0 \leq \ell < r$, and $b \geq 1$ be an integer such that \emph{gcd}$(b,n)=1$. Let $\mu=\frac{k}{r}$. Consider the following sets of elements of $\mathbb{F}_q:$
	$$L=\{ \alpha^i \ | \ i \bmod (r+1)=\ell \}~,$$
	and
	$$D=\{ \alpha^{j+sb} \ | \ s=0,1,\ldots, n-\mu (r+1) \}~,$$
	where $\alpha^j \in L$. Then the cyclic code with the defining set $L \cup D$ is an optimal $(n,k,r)$ cyclic LRC code over $\mathbb{F}_q$.
\end{theorem}
Here the set $L$ determines the locality and the set $D$ ensures the minimum distance of the code.

\begin{theorem}\cite{YM94,MS1977} \label{th2}
Let $C$ be a cyclic code of length $n$ over $\mathbb{F}_q$ with the generator polynomial $g(x)$. Then $C$ is an LCD code if one of the following conditions holds:\\
		$1.$ $g(x)$ is self reciprocal.\\
		$2.$ Inverse of every zero of $g(x)$ is also a zero of $g(x)$.\\
		$3.$ $q^{\ell} \equiv -1 \bmod n$ for some positive integer $\ell$.
\end{theorem}

Let $q$ be a prime power and $n,m$ be two positive integers such that $n | (q^m-1)$. Then the $q$-cyclotomic coset of $a$ modulo $n$, where $a \in \{ 0,1, \ldots, n-1\}$, is defined as
$$ [a]= \{a, aq, aq^2, \ldots, aq^{r-1} \}~,$$
where $r$ is the smallest positive integer such that $aq^r \equiv a \bmod n$.
Following results are straightforward and will be useful for our discussion later.
\begin{prop}
	Let $a \in \{ 0,1, \ldots, n-1\}$ and let $[a]$ represent the $q$-cyclotomic coset of $a$ modulo $n$. Then for any $b \in [a]$, $n-b \in [n-a]$.
\end{prop}

\begin{prop}
	Let $a \in \{ 0,1, \ldots, n-1\}$. Then the cardinality of $[a]$ is a divisor of $\emph{Ord}_n(q)$. Further, $|[1]|=\emph{Ord}_n(q)$.
\end{prop}

\section{Main results}

\subsection{Binary LRC-LCD codes}
Goparaju et al. have presented some constructions of binary cyclic LRC codes in \cite{GC2014}. In this subsection we determine some conditions for these codes to be LCD codes.
\begin{construction}[\cite{GC2014}]
	Let $n = 2^m-1$, $r + 1$ be a factor of $n$, and $\alpha$ be a primitive element of $\mathbb{F}_{2^m}$. Let $C$ be a binary cyclic code of length $n$ with the generator polynomial $g(x)$ having zeroes $\alpha^{j(r+1)}$, where $j$ ranges from $0$ to $\left(\frac{n}{r+1}\right)-1$. Then $C$ is an LRC code with locality $r$ and dimension $k = \frac{rn}{r + 1}$.	
\end{construction}

\begin{theorem}
	Code $C$ constructed in Construction $1$ is an LCD code.	
\end{theorem}
\begin{proof}
	Let $Z$ be the set of zeros of $C$, i.e.,
	$$Z=\left\{\alpha^{j(r+1)} \ | \ j=0,1,\ldots, \left(\frac{n}{r+1}-1 \right)\right\}.$$
	Let $0 \leq i \leq \left( \frac{n}{r+1} -1 \right)$. Then $\alpha^{i(r+1)} \in Z$. Since $(r+1)|n$,
	$$\alpha^{-i(r+1)}=\alpha^{n-i(r+1)} \in Z.$$
    From Theorem \ref{th2}, $C$ is a reversible cyclic code and hence an LCD code.
\end{proof}

Construction $1$ gives distance optimal binary cyclic LRC codes with the minimum distance $2$. To improve the minimum distance, more zeros can be added to $Z$, and the generator polynomial be modified accordingly. It is well known that if $\alpha$ is a zero of $g(x)$ then every element of the cyclotomic coset of $\alpha$ is a zero of $g(x)$. The inverses of the elements of a cyclotomic coset form another cyclotomic coset modulo $n$, which may be the same coset as the original cyclotomic coset. Hence to maintain reversibility of the code, while adding some cyclotomic coset to $Z$, we need to also add to $Z$ the corresponding cyclotomic coset of inverse elements.

\begin{theorem}[Construction $2$]
	Let $n = 2^m-1$, $(r + 1)\ |\ n$ and $\alpha$ be a primitive element of $\mathbb{F}_{2^m}$. Let $g(x)$ be the  generator polynomial of a cyclic code $C$ having the set of zeroes
	$$Z=\left\{\alpha^{j(r+1)} \ | \ j=0,1,\ldots, \left(\frac{n}{r+1}-1 \right)\right\} \cup [\alpha] \cup [\alpha^{n-1}]~.$$
	Then $C$ is an LCD code and also an LRC with locality $r$, minimum distance $d \geq 6$ and the dimension $$k = \frac{r}{r+1}(2^m-1)-2m~.$$	
	
\end{theorem}
\begin{proof}
	Clearly $C$ is an LRC code by Construction $1$. For $n=2^m-1$, $m>1$, $\alpha$ and $\alpha^{-1}$ are in two distinct cyclotomic cosets modulo $n$. Since $Z$ contains $\alpha^{n-2}, \alpha^{n-1}, \alpha^0, \alpha^1, \alpha^2$ consecutive roots, the minimum distance of $C$ satisfies $d \geq 6$. The dimension of $C$ is
	$$k=n-|z| = n- \left( \frac{n}{r+1} + 2m \right) = \frac{r}{r+1}(2^m-1)-2m~.$$
	Further by Theorem \ref{th2}, $C$ is an LCD code since for every zero $\beta$ of $g(x)$, $\beta^{-1}$ is also a zero of $g(x)$.
\end{proof}

\begin{remark*}
	\begin{enumerate}
		\item 	If $r=2$, then $Z$ contains consecutive roots $\alpha^{n-4}, \alpha^{n-3}, \alpha^{n-2}, \alpha^{n-1}, \alpha^0, \alpha^1, \alpha^2, \alpha^3, \alpha^4$ of $g(x)$. Hence the minimum distance of the code $C$ is at least $10$ and its dimension is $$k=\frac{2}{3}(2^m-1)-2m~.$$
		\item The minimum distance of the code can further be improved by adding suitable cyclotomic cosets. If $r>2$ and $Z$ contains $\alpha^{j(r+1)}$, where $j=0,1,\ldots, \left( \frac{n}{r+1}-1 \right)$, and cyclotomic cosets $[\alpha], [\alpha^{n-1}], [\alpha^3]$ and $[\alpha^{n-3}]$, then the minimum distance of $C$ is at least $10$ and its dimension is
		$$k=n-|z| \geq n- \left( \frac{n}{r+1} + 4m \right) = \frac{r}{r+1}(2^m-1)-4m~.$$
	\end{enumerate}
\end{remark*}

\begin{example}
	Let $n=2^6-1$, $m=6$ and $r=2$. Consider the set
	$$Z=\left\{\alpha^{3j} \ | \ j=0,1,\ldots, 20 \right\} \cup [\alpha] \cup [\alpha^{62}]~.$$
	Let $C$ be the corresponding cyclic code. Clearly $Z$ contains $9$ consecutive roots of $C$ from $\alpha^{-4}$ to $\alpha^{4}$. Hence the minimum distance of $C$ is at least $10$ and the dimension is $k=30$.
\end{example}

\begin{example}
	Let $n=2^8-1$, $m=8$ and $r=4$. Consider the set
	$$Z=\left\{\alpha^{5j} \ | \ j=0,1,\ldots, 50 \right\} \cup [\alpha] \cup [\alpha^{254}] \cup [\alpha^3] \cup [\alpha^{252}]~.$$
	Let $C$ be the corresponding cyclic code. Clearly $Z$ contains $13$ consecutive roots of $C$ from $\alpha^{-6}$ to $\alpha^{6}$. Hence the minimum distance of $C$ is at least $14$ and the dimension is $k=172$.
\end{example}

\subsection{LRC-LCD codes over $\mathbb{F}_q$}

Now we consider constructions of LRC-LCD cyclic codes over an arbitrary finite field  $\mathbb{F}_q$ such that the code length $n$ divides $q-1$. Then $\mathbb{F}_q$  contains a primitive $n$th root of unity, and the size of each $q$-cyclotomic coset modulo $n$ is equal to one.

\begin{theorem} \label{th3}
	Let $n|(q-1)$ and $\alpha$ be a primitive $n$-th root of unity. Let $k$ and $r$ be integers such that $1 \leq r \leq k$, $(r+1)|n$ and $r|k$. Suppose $\left( \frac{n}{r+1}-\frac{k}{r}\right)$ is even. Consider two sets
	$$L=\{ \alpha^i \ | \ i \bmod(r+1) =0 \}~,$$
	and
	$$D=\{\alpha^s \ | \ s=-t, -t+1, \ldots, -1, 0, 1, \ldots, t-1, t\}~,$$
	where $t=\frac{1}{2}\left( n - \frac{k(r+1)}{r} \right)$. If $C$ is the cyclic code of length $n$ over $\mathbb{F}_q$ with the defining set $Z=L \cup D$, then $C$ is an optimal LRC code and also an LCD code.
	\end {theorem}
	
	\begin{proof}
		First we will prove that $C$ is an LCD code.
		From Theorem \ref{th2}, $C$ is an LCD code if and only if $\beta^{-1}$ is also a root of $g(x)$ whenever $\beta$ is a root of $g(x)$, where $g(x)$ is the generating polynomial of $C$.
		
		Let $\alpha^i \in Z=L \cup D$. If $\alpha^i \in D$ then clearly $\alpha^{-i} \in D$. If $\alpha^i \in L$ then $i \bmod (r+1)=0$. Therefore $\alpha^{n-i}=\alpha^{-i} \in L$, as $(n-i) \bmod (r+1) = 0$. So $C$ is an LCD code.
		
		Now we will prove that $C$ is an optimal LRC code. From Theorem \ref{th1}, the set $L$ determines the locality of code $C$ when $\ell=0$, i.e., for the present case. Dimension of $C$ is
\begin{align*}
 n - &|Z| = n-(|L|+|D|-|L \cap D|) \\
          &= n- \frac{n}{r+1}-(2t+1)+\left(2 \left \lfloor \frac{t}{r+1} \right \rfloor+1 \right) \\
	&= \frac{nr}{r+1}-2 \left( t- \left \lfloor \frac{t}{r+1}  \right \rfloor \right) \\
	&= \frac{nr}{r+1}-2\left( \frac{1}{2}\left( n - \frac{k(r+1)}{r} \right) - \left \lfloor \frac{1}{2}\left( \frac{n}{r+1} - \frac{k}{r} \right)  \right \rfloor \right)\\
	&= \frac{nr}{r+1} - \left( \frac{nr}{r+1}-k \right) = k~.
\end{align*}

Since $D$ contains $2t+1$ consecutive roots, the minimum distance $d$ of $C$ satisfies
		\begin{align*}
			d & \geq 2t+2 \\
			& = n- \frac{k(r+1)}{r}+2 \\
			&= n-k-\frac{k}{r}+2~.
		\end{align*}
	\end{proof}

	\begin{example}
		Let $q=37$, $n=36$ and $r=5$. Let $\alpha$ be a primitive $n$-th root of unity.
		Then we have $$L=\{ \alpha^0, \alpha^6, \alpha^{12}, \alpha^{18}, \alpha^{24}, \alpha^{30} \}~.$$
		For $k=20$, $\left( \frac{n}{r+1}-\frac{k}{r} \right)=2$ is even. So
		$$t=\frac{1}{2}\left( n- \frac{k(r+1)}{r} \right)=6.$$
		Hence we get $D=\{\alpha^{-6}, \alpha^{-5}, \alpha^{-4},\alpha^{-3}, \alpha^{-2}, \alpha^{-1}, \alpha^{0}, \alpha^{1}, \\ \alpha^{2}, \alpha^{3}, \alpha^{4}, \alpha^{5}, \alpha^{6}\}.$
		If $C$ is the cyclic code with the defining set $Z=L \cup D$, then $C$ is an LRC-LCD cyclic code of length $n=36$, dimension $k=20$, locality $r=5$ and the minimum distance $d \geq 14$.
		From (\ref{eq1}) we have
		$$d \leq n-k-\left \lceil \frac{k}{r} \right \rceil +2 = 14~.$$
		Therefore $d = 14$, and hence $C$ is an optimal LRC code.
	\end{example}

	In the next result we propose a construction of cyclic LRC-LCD codes of length $n$ over $\mathbb{F}_q$ without imposing the condition $r|k$. In some cases, codes are still distance optimal but in other cases minimum distance is reduced by one from the optimal value.
	\begin{theorem}\label{th4}
		Let $n|(q-1)$ and $\alpha$ be a primitive $n$-th root of unity. Let $k$ and $r$ be integers such that $1 \leq r \leq k$ and $(r+1)|n$. Let $t=\left\lfloor \frac{nr-k(r+1)}{2r}\right\rfloor$, and suppose $2r\bigm|\left( \frac{nr}{r+1}-k-2a \right)$, where $a=t \bmod (r+1)$. Consider two sets
		$$L=\{ \alpha^i \ | \ i \bmod(r+1) =0 \}$$
		and
		$$D=\{\alpha^s \ | \ s=-t, -t+1, \ldots, -1, 0, 1, \ldots, t-1, t\}.$$
		If $C$ is the cyclic code with the defining set $Z=L \cup D$, then $C$ is an LRC-LCD cyclic code of length $n$  over $\mathbb{F}_q$ with the minimum distance
		$$d \geq n-k-\left \lceil \frac{k}{r} \right \rceil +1~.$$
		Further, $C$ is an optimal LRC code for $a=0$ and $a=r$.
		\end {theorem}
		
		\begin{proof}
						It can easily be shown that $C$ is an LCD code by the same argument as in Theorem \ref{th3}.
			From Theorem \ref{th1}, the set $L$ gurantees the locality of the code $C$ to be $r$ for $\ell=0$.
			Dimension of $C$ is
\begin{align}\label{eq2}
 n-|Z| &= n-(|L|+|D|-|L \cap D|)  \nonumber \\
	&= \frac{nr}{r+1}-2 \left( t- \left \lfloor \frac{t}{r+1}  \right \rfloor \right)~. 
\end{align}
Now 
\begin{equation} \label{eq3}
t- \left \lfloor \frac{t}{r+1} \right \rfloor = \left \lfloor \frac{n}{2} - \frac{k(r+1)}{2r} \right \rfloor- \left \lfloor \frac{n}{2(r+1)} - \frac{k}{2r} \right \rfloor~.
\end{equation}
For some real numbers $x$ and $y$, if $x-y$ is an integer then $\lfloor x \rfloor - \lfloor y \rfloor=\lfloor x-y \rfloor$ . Since $2r\bigm|\left( \frac{nr}{r+1}-k-2a \right)$,
\begin{align*}
\frac{n}{2} - \frac{k(r+1)}{2r}&- \left( \frac{n}{2(r+1)} - \frac{k}{2r} \right) = \left( \frac{nr}{2(r+1)} - \frac{k}{2} \right) \\
&= \frac{1}{2} \left( \frac{nr}{2(r+1)} -k-2a \right)+a \\
&= r \left( \frac{1}{2r} \left( \frac{nr}{2(r+1)} -k-2a \right)  \right) +a
\end{align*}
is an integer. Hence from equation \eqref{eq3}, 
$t- \left \lfloor \frac{t}{r+1} \right \rfloor = \frac{1}{2} \left( \frac{nr}{r+1} - k \right).$  and from equation \eqref{eq2}, dimension of $C$ is $k$.
Since the set $D$ contains $2t+1$ consecutive roots of the generator polynomial $g(x)$ of $C$, we get
			\begin{eqnarray*}
				d &\geq& 2t+2 \\
				&=& 2 \left \lfloor \frac{nr-k(r+1)}{2r} \right \rfloor+2 \\
				&=& 2 \left \lfloor \frac{(r+1)}{2r} \left( \frac{nr}{r+1} -k \right)\right \rfloor+2 \\
				&=&  2 \left \lfloor \frac{(r+1)}{2r} \left( \frac{nr}{r+1} -k -2a \right) + \frac{a(r+1)}{r} \right \rfloor+2~.
			\end{eqnarray*}
			Since $2r\bigm|\left( \frac{nr}{r+1}-k-2a \right)$, $\frac{(r+1)}{2r} \left( \frac{nr}{r+1} -k -2a \right) $ is an integer. Hence
			
			\begin{eqnarray*}
				d &\geq& n- \frac{k(r+1)}{r}-\frac{2a(r+1)}{r} +2 \left \lfloor \frac{a(r+1)}{r} \right\rfloor +2 \\
				&=& n-k-\frac{k}{r} - 2\left\{ \frac{a(r+1)}{r} - \left \lfloor \frac{a(r+1)}{r} \right\rfloor \right\} +2\\
				&\geq& n-k-\left \lceil \frac{k}{r} \right \rceil - 2\left\{ \frac{a(r+1)}{r} - \left \lfloor \frac{a(r+1)}{r} \right\rfloor \right\} +2~.
			\end{eqnarray*}
			\underline{For $a=0$ and $a=r$},
			$$\frac{a(r+1)}{r} - \left \lfloor \frac{a(r+1)}{r} \right\rfloor=0~,$$
			which implies that
			$$d\geq n-k-\left \lceil \frac{k}{r} \right \rceil +2~.$$
			From Theorem $2.1$, $C$ is an optimal LRC code.
			\flushleft
			\underline{For $0<a<r$},
			$$\frac{a(r+1)}{r} - \left \lfloor \frac{a(r+1)}{r} \right\rfloor <1~,$$
			which implies that
			$$
			d>n-k-\left \lceil \frac{k}{r} \right \rceil \geq  n-k- \left
			\lceil \frac{k}{r} \right \rceil +1~.
			$$
		\end{proof}
		
		\begin{remark*}
			For $a=0$, Theorem \ref{th4} reduces to Theorem \ref{th3}.
		\end{remark*}
		
		\begin{example}
			Let $q=17$, $n=16$ and $r=3$. Let $\alpha$ be a primitive $n$-th root of unity.
			Then we have $$L=\{ \alpha^0, \alpha^4, \alpha^{8}, \alpha^{12} \}~.$$
			For $k=8$, $$t=\left \lfloor \frac{nr-k(r+1)}{2r} \right \rfloor =2~.$$
			and $a= t \bmod (r+1) =2$. Since $\left( \frac{nr}{r+1} -k-2a \right)=0$ is divisible by $2r$, take the set $D$ as
			$$D=\{\alpha^{-2}, \alpha^{-1}, \alpha^0, \alpha^1,\alpha^2 \}~.$$
			
			If $C$ is the cyclic code with the defining set $Z=L \cup D$, then $C$ is an LCD code and also LRC code of length $n=16$, dimension $k=8$, locality $r=3$ and the minimum distance $d \geq 6$.
			It satisfies the bound on the minimum distance given in Theorem \ref{th4}, i.e.,
			$$d \geq n-k-\left \lceil \frac{k}{r} \right \rceil +1 = 6~.$$
		\end{example}
		
		In the following example, it is shown that for some value of $k$, we get optimal cyclic LRC codes and for some value of $k$ the minimum distance is one less than the optimal value.
		
		\begin{example}
			Let $q=67$, $n=66$ and $r=5$. Let $\alpha$ be a primitive $n$-th root of unity.
			Then we have $$L=\{ \alpha^0, \alpha^{6}, \alpha^{12}, \alpha^{18}, \alpha^{24}, \alpha^{30}, \alpha^{36}, \alpha^{42}, \alpha^{48}, \alpha^{54}, \alpha^{60} \}~.$$
			\underline{For $k=35$}, \\
			$$t=\left \lfloor \frac{nr-k(r+1)}{2r} \right \rfloor =12~,$$
			and $a=0$. Then $\left( \frac{nr}{r+1} -k-2a \right)=20$, which is divisible by $2r$. Take the set $D$ as 
$$D=\{\alpha^{-12}, \alpha^{-11}, \ldots, \alpha^{-1}, \alpha^{0}, \alpha^{1}, \ldots, \alpha^{11}, \alpha^{12}\}~.$$
			If $C$ is the cyclic code with the defining set $Z=L \cup D$, then $C$ is an LRC-LCD code of length $n=66$, dimension $k=35$, locality $r=5$ and the minimum distance $d \geq 26$. Further $d$ satisfies the upper bound given in (\ref{eq1}), i.e., $d \leq n-k-\left \lceil \frac{k}{r} \right \rceil +2 = 26$,
			Therefore $C$ is an optimal LRC code.  \\

			\noindent \underline{For $k=37$}, $$t=\left \lfloor \frac{nr-k(r+1)}{2r} \right \rfloor =10~,$$
			and $a= t \bmod (r+1) =4$. Then $\left( \frac{nr}{r+1} -k-2a \right)=10$, which is divisible by $2r$. Take the set $D$ as
			$$D=\{\alpha^{-10}, \alpha^{-9}, \ldots, \alpha^{-1}, \alpha^{0}, \alpha^{1}, \ldots, \alpha^{9}, \alpha^{10}\}~.$$
			If $C$ is the cyclic code with the defining set $Z=L \cup D$, then $C$ is an LRC-LCD code of length $n=66$, dimension $k=37$, locality $r=5$ and the minimum distance $d \geq 22$. It satisfies the bound on the minimum distance given in Theorem \ref{th4}, i.e.,
			$$d \geq n-k-\left \lceil \frac{k}{r} \right \rceil +1 = 22.$$
		\end{example}

\section{Conclusion}
In this paper, we have established a connection between LRC codes and LCD codes. We obtained some conditions on the construction of cyclic LRC codes which are also LCD (reversible) codes. We believe that the LRC-LCD codes constructed here would be useful for distributed data storage systems.

\bibliography{LRC-LCD}
\bibliographystyle{IEEEtran}

\end{document}